\documentclass[journal]{IEEEtran}

\usepackage{amsmath,amssymb}

\newtheorem{theorem}{Theorem}

\newtheorem{definition}{Definition}

\newtheorem{lemma}{Lemma}
\newtheorem{remark}{Remark}
\newenvironment{proof}[1][Proof]{\noindent\textbf{#1.} }{\ \rule{0.5em}{0.5em}}

\ifCLASSINFOpdf
\else
\fi

\def\la{\langle}
\def\ra{\rangle}
\def\a{\alpha}
\def\b{\beta}

\def\tr{{\rm tr}}

\usepackage{epsfig}

\usepackage{color}

\begin{document}

\title{Recovering quantum information through partial access to the
environment}
\author{Laleh Memarzadeh,
        Chiara Macchiavello
        and Stefano Mancini% <-this % stops a space
\thanks{L. Memarzadeh is with the Department of Physics, Sharif University of Technology,
Teheran, Iran.}% <-this % stops a space
\thanks{C. Macchiavello is with the Department of Physics ``A. Volta", University of Pavia, I-27100 Pavia, Italy}
\thanks{S. Mancini is with the School of Science and Technology, University of Camerino,
I-62032 Camerino, Italy.}% <-this % stops a space
}

% The paper headers
%\markboth{Journal of \LaTeX\ Class Files,~Vol.~, No.~, January~2011}%
%{Shell \MakeLowercase{\textit{et al.}}: Bare Demo of IEEEtran.cls for Journals}

\maketitle

\begin{abstract}
We investigate the possibility of correcting errors occurring on
a multipartite system through a feedback mechanism that acquires
information from partial access to the environment.
A partial control scheme of this kind might be useful when dealing with
correlated errors. In fact, in such a case, it could be enough to gather local
information to decide what kind of global recovery to perform.
Then, we apply this scheme to the depolarizing and correlated errors,
and quantify its
performance by means of the entanglement fidelity.
\end{abstract}

\begin{IEEEkeywords}
Quantum feedback, Quantum channels, Quantum correlated errors.
\end{IEEEkeywords}

\IEEEpeerreviewmaketitle

\section{Introduction}

\IEEEPARstart{Q}  {uantum noise}  is
 {the main obstacle for realizing quantum information tasks.
It results from the errors introduced on the system's state by the unavoidable interaction with the surrounding environment \cite{Zurek}.
As a consequence the quantum coherence features of the system's state are washed out.
To restore them, one could naively think of measuring the system
(gathering information about its state) and then applying a correction
procedure. This is the idea underlying the quantum feedback control mechanism \cite{WisemanBook}.
Actually, also quantum error correcting codes can be thought as belonging to this kind of strategy \cite{Knill}.
In particular,
one can make a measurement on the final state of the environment
and consider its classical result to recognize what kind of error has
occurred on the system due to the interaction with the environment.
Then, a proper correction should be performed on the system to reduce the effect of quantum noise \cite{QLF}.
Recently, a lot of attention has been devoted to this scheme from different aspects.
In \cite{HaydenKing,SVW} the capacity for this scenario has been studied and in
\cite{BCD} it has been shown that in certain cases repeated application of this scheme allows to remove completely the effects of quantum noise.
For a given measurement the optimal recovery scheme (the recovery necessary to restore the maximum value of quantum information) has been derived \cite{QLF}, while in \cite{MCM} it has been shown that the optimal measurement depends on the dimension of the system's Hilbert space.

In extending this quantum control strategy to multipartite systems, we must deal with a more intricate scenario. For instance
access to all subsystems' environments may not be available. Then we will
address the problem of recovering quantum information by feedback partial
control, that is the measurement is only done on some of the subsystems' environments while the actuation is performed on all subsystems.
In this case the feedback scheme will be effective if errors occurring on different subsystems are somehow correlated, so that gaining information on the measured subsystems also means to indirectly gain information about non-measured ones.
This will help in designing the recovery operation on the whole system.
We will consider a quite general kind of correlated errors on qubits and
determine the optimal recovery depending on the degree of errors' correlation.
We will also find the scaling of the performance versus the number of qubits (subsystems) while monitoring the error just on one
of them.

The layout of the paper is as follows. In Section II we briefly present the main conceptual and computational tools
needed to recover quantum information by means of a quantum feedback control scheme. To get some insights we apply,  in Section III, this strategy to the correlated depolarizing channel for two qubits when only one is monitored.
We then derive the main result for the system of $n$ qubits in Section IV and we draw our conclusions in
Section V.

%%%%%%%%%%%%%%%%%%%%%%%%%%%%%%%%%%%

\section{Recovering Quantum Information by Feedback Control}

The evolution of a system interacting with an environment can be
described by a completely positive and trace preserving map $T:\mathcal{L}(\mathcal{H}_{initial})\rightarrow \mathcal{L}(\mathcal{H}_{final})$
transforming the initial system's density operator in Hilbert space $\mathcal{H}_{initial}$ to a final density operator in Hilbert
space $\mathcal{H}_{final}$ ($\mathcal{L}(\mathcal{H})$ is the space of linear operators on $\mathcal{H}$). At the same time the initial state of the environment in
Hilbert space $\mathcal{K}_{initial}$ is mapped into a final one in $\mathcal{K}_{final}$. The evolution of the system can be described as the unitary
evolution of system and environment given by the unitary operator
$U: \mathcal{H}_{initial}\otimes \mathcal{K}_{initial}\rightarrow \mathcal{H}_{final}\otimes \mathcal{K}_{final}$.
By denoting with  $\rho$ and $\sigma$ the initial state of the system and environment respectively, the map of the
system evolution reads
\begin{equation}
T(\rho)=\tr_{_{{\mathcal{K}}_{final}}}[U(\rho\otimes\sigma) U^{\dag}],
\nonumber
\end{equation}
where $\tr_\bullet$ denotes the trace on the space $\bullet$.

To acquire some information about the errors occurred on the system one can perform
a measurement on the environment after the interaction with the system
has taken place. In general, this is described by a Positive Operator
Valued Measure (POVM) on $\mathcal{K}_{final}$, namely a set of
operators $M_{\alpha}\in \mathcal{L}(\mathcal{K}_{final})$ satisfying
\begin{equation}\label{measurment}
\sum_{\alpha}M_{\alpha}=I, \hskip 1 cm M_{\alpha}>0.
\end{equation}
The index $\alpha$ labels the classical measurement outcomes. Considering
an arbitrary observable $A\in \mathcal{L}(\mathcal{H}_{final})$, the expectation value of this observable is
\begin{equation}\label{Aave}
<A>=\tr_{_{\mathcal{H}_{final}}}\tr_{_{\mathcal{K}_{final}}}[ U(\rho\otimes\sigma )U^{\dag} (A\otimes I)],
\end{equation}
where $I$ is the identity on $\mathcal{L}(\mathcal{K}_{final})$.

\begin{definition}
We define by $T_{\alpha}:\mathcal{L}(\mathcal{H}_{initial})\rightarrow \mathcal{L}(\mathcal{H}_{final})$,
\begin{equation}
T_{\alpha}(\rho):=\tr_{_{\mathcal{K}_{final}}}[U(\rho\otimes\sigma )U^{\dag} (I\otimes  M_{\alpha} )],\nonumber
\end{equation}
the selected channel output corresponding to the outcome $\alpha$.
\end{definition}\medskip

Then, replacing $I$ in \eqref{Aave} with the identity resolution (\ref{measurment}), we get
\begin{equation}
<A>=\sum_{\alpha}\tr_{_{\mathcal{H}_{final}}}(T_{\alpha}(\rho)A).\nonumber
\end{equation}
Rewriting the expectation value of $A$ in the following way
\begin{equation}
<A>=\sum_{\alpha}p_{\alpha}\,\tr\left[\frac{T_{\alpha}(\rho)}{p_{\alpha}}A\right],\nonumber
\end{equation}
we can conclude that $p_{\alpha}=\tr(T_{\alpha}(\rho))$ is the probability of getting
$\alpha$ as the result of the measurement and the density matrix
$\frac{1}{p_{\alpha}}T_{\alpha}(\rho)$ as the selected state of the system after performing the measurement on the environment.

We can also define the most informative measurement \cite{QLF} in terms of Kraus operators \cite{Kraus} composing the channel $T_\a$.

\begin{definition}\label{mostinf}
Given a channel $T =\sum_{\alpha}T_{\alpha}$, the most informative measurement on the environment,
is such that we can describe the selected output of
the channel $T_{\a}$ by a single Kraus operator $T_{\a}(\rho)=t_{\a}\rho t_{\a}^{\dag}$.
\end{definition}\medskip

Therefore
 \begin{equation}\label{Kraus}
 T=\sum_{\alpha}t_{\a}\rho t_{\a}^{\dag}, \hskip 1cm \sum_{\a}t_{\a}^{\dag}t_{\a}=I.
  \end{equation}
In order to correct the errors due to the interaction with the environment,
we have to introduce a recovery operation.

\begin{definition}
Let $R_{\a}: \mathcal{L}(\mathcal{H}_{final})\rightarrow \mathcal{L}(\mathcal{H}_{initial})$ be the recovery operator
that acts on the selected output of the channel $T_{\a}(\rho)$ and
depends on the classical outcome of the measurement $\a$.
Then, the overall corrected channel takes the form
\begin{equation}
\label{Tcorr}
T_{corr}:=\sum_{\a}R_{\a}\circ T_{\a}.
\end{equation}
\end{definition}\medskip
Using (\ref{Kraus}) and a Kraus representation \cite{Kraus} for the recovery channel $R_{\a}$
\begin{equation}
R_{\a}(\rho')=\sum_{\beta}r_{\b}^{(\a)}\rho'r_{\b}^{(\a)\dag},\hskip 1 cm
\sum_{\beta}r_{\b}^{(\a\dag)}r_{\b}^{(\a)}=I,
\nonumber
\end{equation}
we can decompose the corrected channel as
\begin{equation}\label{krausTcorr}
T_{corr}(\rho)=\sum_{\a,\b}r_{\b}^{(\a)}t_{_{\a}}\rho \,t_{_{\a}}^{\dag}r_{\b}^{(\a) \dag}.
\end{equation}

To quantify the performance of the correction
scheme, we use the entanglement fidelity \cite{EntFid,Nielsen}.
\begin{definition}\label{def:entfid}
For a general map $\Phi: \mathcal{L}(\mathcal{H})\rightarrow \mathcal{L}(\mathcal{H})$ with Kraus operators $A_k$, the entanglement fidelity is defined as
\begin{equation}
\label{efid}
F(\Phi):=\la\Psi| \Phi\otimes I(|\Psi\ra\la\Psi|)|\Psi\ra=\frac{1}{d^2}\sum_k|\tr (A_k)|^2,
\end{equation}
where $d=dim \mathcal{H}$ and $|\Psi\ra\in \mathcal{H}\otimes \mathcal{H}$ is a maximally entangled state.
\end{definition}\medskip

We are interested in $F(T_{corr})$, the entanglement fidelity of the corrected map \eqref{Tcorr}. As a consequence
of \eqref{krausTcorr} and \eqref{def:entfid} we have
\begin{equation}\label{ftc}
F(T_{corr})=\frac{1}{d^2}\sum_{\a,\b}|\tr(r_{\b}^{(\a)}t_{_{\a}})|^2.
\end{equation}
The entanglement fidelity reaches its maximum value if quantum information is completely recovered,
or in other words if the corrected channel becomes an identity map. In \cite{QLF} it has been shown that
there exists a family of operators that
completely recover quantum information if and only if
\begin{equation}\label{correctable}
t_{\a}^{\dag}t_{\a}=c_{\a}I,\hskip 5mm \forall \a,
\end{equation}
with $c_{\a}\in\mathbb{R}_+$ and $\sum_{\a}c_{\a}=1$.

These results are obtained with the assumption that full access to the
environment is available and it is possible
to perform a measurement on the whole environment after the interaction
with the system. However, more generally we should assume that our access to the environment is partial. Here
we want to investigate how the performance of this correction scheme behaves
in this case and to see if we can still completely retrieve quantum
information.
To shed light on this problem, we study a map for which the complete
recovery of quantum information
is possible, provided that we have complete access to the environment.

Specifically we are going to consider the depolarizing quantum channel.
In the following we will consider
\begin{eqnarray}
\label{spaces}
\mathcal{H}:=\mathbb{C}^2,\quad
\mathcal{K}:=\mathbb{C}^2\otimes \mathbb{C}^2.\nonumber
\end{eqnarray}

\begin{definition}
The single qubit depolarizing channel is defined by
\begin{eqnarray}
\mathcal{H}_{initial}=\mathcal{H}_{final}&=&\mathcal{H},\nonumber\\
\mathcal{K}_{initial}=\mathcal{K}_{final}&=&\mathcal{K},\nonumber
\label{spaces1}
\end{eqnarray}
and
\begin{equation}
\label{depuni}
t_\alpha=\sqrt{p_\alpha}\sigma_\alpha,\nonumber
\end{equation}
where the operators $\sigma_\alpha$, with $\alpha=0,1,2,3$
(the Greek indices go from 0 to 3 while Latin indices go from 1 to 3), denote the Pauli operators (including the
identity operator), while $p_{0}=1-p$ and $p_1=p_2=p_3=\frac{p}{3}$.
\end{definition}\medskip

\begin{remark}
Since the Pauli operators satisfy the condition
(\ref{correctable}), the quantum information in this case
can be completely recovered.
To achieve this it is enough to consider a recovery channel described
by a single Kraus operator $\sigma_\a$, where $\a$ is the classical outcome
of the measurement.
Hence from \eqref{ftc} we have
\begin{equation}
\label{18}
F(T_{corr})=\sum_{\alpha=0}^3 \left| \sqrt{p_\alpha} \right|^2=1.\nonumber
\end{equation}
\end{remark}

However, the situation will be different when we enlarge the Hilbert spaces of the system and environment while performing a  measurement just on
a subsystem of the environment. In the next sections we show how the
performance of this scheme behaves when our access to the environment is partial.

%%%%%%%%%%%%%%%%%%%%%%%%%%%%%%%%%%%%%%%%%%

\section{depolarizing channel for two qubits}

To study the feedback control scheme with partial access to the environment, we start by analyzing the
depolarizing channel
$T:\mathcal{L}(\mathcal{H}^{\otimes 2})\rightarrow \mathcal{L}(\mathcal{H}^{\otimes 2})$ acting
 on two qubits. In the following we assume that we can perform a measurement
on $\mathcal{L}(\mathcal{K})$ while the state of the environment belongs to $\mathcal{L}(\mathcal{K}^{\otimes 2})$.
Since the access to the environments is partial, the measurement can not be a most informative measurement
(see Definition \ref{mostinf}) and therefore the
selected output of the channel is in general given by
\begin{equation}
T_{\a}(\rho)=\sum_{\beta}t_{_{\a,\beta}}\rho t_{_{\a,\beta}}^{\dag}.\nonumber
\end{equation}
The Kraus operators $t_{\a,\b}$ will be
\begin{equation}
t_{\a,\b}=\sqrt{p_{\a}p_{\b}}\sigma_{\alpha}\otimes \sigma_{\b},\nonumber
\end{equation}
in the case of no correlations and
\begin{equation}
t_{\a,\b}=\sqrt{p_{\a}}\delta_{_{\a,\b}}\sigma_{\alpha}\otimes \sigma_{\b},\nonumber
\end{equation}
in the case of perfect correlations. In the first case
the outcome of the measurement does not give any information about the error occurred on the second qubit,
therefore the selected output of the channel is
\begin{equation}
\label{Tuc}
T_{\alpha}^{uc}=\sum_{\b}p_{\a}p_{\b}(\sigma_{\alpha}\otimes \sigma_{\beta})\rho(\sigma_{\alpha}\otimes \sigma_{\beta})^{\dag}.
\end{equation}
In the second case, the measurement of the environment is the most informative one (according to Definition \ref{mostinf}),
hence the selected output of the channel is
\begin{equation}
\label{Tcc}
T_{\a}^{cc}(\rho)=p_{\a} (\sigma_{\a}\otimes \sigma_{\a})\rho(\sigma_{\a}\otimes \sigma_{\a})^{\dag}.
\end{equation}
We will now consider a more general situation, which interpolates
between the two above situations. More explicitly, we consider a correlated
noise model which is a convex combination of two cases described above, namely
the uncorrelated noise and the completely correlated noise for two qubits
\cite{MP}.

\begin{definition}\label{T2}
The following convex combination of channels \eqref{Tuc}
and \eqref{Tcc}
\begin{equation}
T_{\a}:=(1-\mu)T_{\a}^{uc}+\mu T_{\a}^{cc},\nonumber
\end{equation}
where $\mu\in [0,1]$ quantifies the amount of correlation in noise,
defines the selected output.
\end{definition}\medskip

Our aim is now to design the recovery channel in order to achieve the maximum
value of the entanglement fidelity for
the corrected channel.

\begin{lemma}\label{th:rec}
The recovery map
\begin{equation}
R_{\a}(\rho'):=\sum_{\gamma}q_{_{\a,\gamma}}(\sigma_{\a}\otimes \sigma_{\gamma})\rho'(\sigma_{\a}\otimes \sigma_{\gamma}),\hskip 5 mm \sum_{\gamma}q_{_{\a,\gamma}}=1,\nonumber
\end{equation}
is optimal for the channel of Definition \ref{T2}.
\end{lemma}

\begin{proof}
The recovery map can be described, without loss of generality, as
 \begin{eqnarray}
R_{\a}(\rho')&=&\sum_{\gamma}(\sigma_{\a}\otimes A_{\gamma}^\a)\rho'(\sigma_{\a}\otimes A_{\gamma}^\a)^{\dag},
\nonumber
\end{eqnarray}
where the single qubit operators $A_{\gamma}^\a$ can be expressed
in terms of the identity and the Pauli operators as
\begin{equation}
A_{\gamma}^\a=\sum_{\delta} c_{\gamma,\delta}^{\a} \sigma_{\delta}.\nonumber
\end{equation}
The completeness condition $\sum_{\gamma}
(\sigma_{\a}\otimes A_{\gamma}^\a)^{\dag}(\sigma_{\a}\otimes A_{\gamma}^\a)=I$
gives the following normalisation condition for the coefficients
$c_{\gamma,\delta}^{\a}$
\begin{eqnarray}
\sum_{\gamma,\delta}|c_{\gamma,\delta}^{\a}|^2=1.\nonumber
\end{eqnarray}
Notice that this is the most general map we can use as recovery.
Actually, on the first qubit the optimal action is to invert the action of
$\sigma_\a$ by $\sigma_\a$ itself, while on the second one we consider a
generic operator $A_{\gamma}^\a$ possibly correlated with the one on the first
qubit (that is the reason for the presence of the index $\a$ on
$A_{\gamma}^\a$).
Then the entanglement fidelity of the corrected channel, using \eqref{ftc},
takes the form
\begin{equation}
 F(T_{corr})=(1-\mu)\sum_{\a,\beta,\gamma}p_\a p_{\beta}|c_{\gamma,\beta}^{\a}|^2
 +\mu \sum_{\a,\gamma}p_\a |c_{\gamma,\a}^{\a}|^2.\nonumber
 \end{equation}
Notice that this can be rewritten as
 \begin{equation}
 \label{fid}
 F(T_{corr})=(1-\mu)\sum_{\a,\beta}p_\a p_{\beta}q_{\a,\beta}
 +\mu \sum_{\a}p_\a q_{\a,\a},
 \end{equation}
 where we defined the probabilities
 \begin{equation}
 q_{\a,\beta}=\sum_{\gamma}|c_{\gamma,\beta}^{\a}|^2.\nonumber
 \end{equation}
Eq. \eqref{fid} is the same expression that we would obtain by assuming
the recovery as a Pauli channel, namely with the Kraus operators
\begin{equation}
\sigma_{\a}\otimes A_{\gamma}^\a=\sqrt{q_{\a,\gamma}} \sigma_{\a}\otimes\sigma_{\gamma}.\nonumber
\end{equation}
with $\sum_{\gamma}q_{\a,\gamma}=1$ for all $\gamma$.
\end{proof}

%%%%%%%%%%%%%%%%%%%%%%%%%%%%%%%%%%%%%%

By virtue of Lemma \ref{th:rec}, the corrected channel can be written as
\begin{eqnarray}
T_{corr}(\rho)&=(1-\mu)\sum_{\a,\b,\gamma}p_{\a}p_{_{\b}}q_{\a,\gamma}( {I}\otimes \sigma_{\gamma}\sigma_{\b})\rho( {I}\otimes \sigma_{\b}\sigma_{\gamma})\cr\cr
&+\mu\sum_{\a,\gamma} p_{\a}q_{\a,\gamma}( {I}\otimes \sigma_{\gamma}\sigma_{\a})\rho ( {I}\otimes \sigma_{\a}\sigma_{\gamma}),\nonumber
\end{eqnarray}
and its entanglement fidelity becomes
\begin{eqnarray}
F(T_{corr})&=&\frac{1-\mu}{16}\sum_{\a,\b,\gamma}p_{\a}p_{_{\b}}q_{_{\a,\gamma}}|\tr( {I}\otimes \sigma_{\gamma}\sigma_{\b})|^2\cr
&+&\frac{\mu}{16}\sum_{\a,\gamma} p_{\a}q_{_{\a,\gamma}}|\tr( {I}\otimes \sigma_{\a}\sigma_{\gamma})|^2\cr\cr
&=&(1-\mu)\sum_{\a,\b}p_{\a}p_{_{\b}}q_{_{\a,\b}}+\mu\sum_{\a} p_{\a}q_{_{\a,\a}}.\nonumber
\end{eqnarray}
Taking into account that $\sum_{\gamma}q_{_{\a,\gamma}}=1$ for all
values of $\a$, the above equation is simplified as follows
\begin{eqnarray}\label{27}
F(T_{corr})&=&(1-\mu)\frac{p}{3}\cr\cr
&+&(1-p)\left((1-\mu)(1-\frac{4p}{3})+\mu\right)q_{_{0,0}}\cr\cr
&+&\frac{p}{3}\sum_{i=1}^3\left((1-\mu)(1-\frac{4p}{3})q_{_{i,0}}+\mu q_{_{i,i}}\right).
\end{eqnarray}
%%%%%%%%%%%%%%%%%%%%%%%%%%%%%%%%%%%%%%%
\begin{figure}[t]
\centering
\includegraphics[height=0.25\textheight]{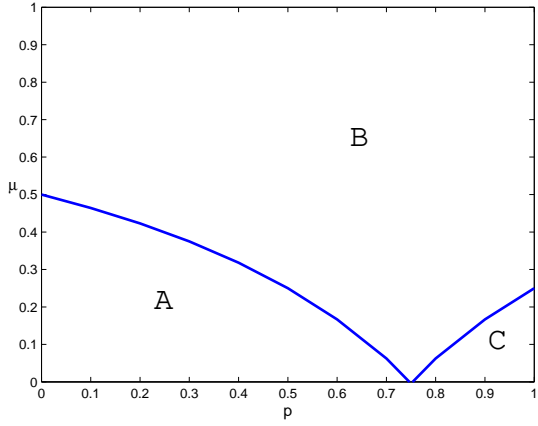}
\caption{Different parameters regions for optimal
recovery in the case of two qubits channel.} \label{Regions}
\end{figure}
%%%%%%%%%%%%%%%%%%%%%%%%%%%%%%%%%%%%%%%%%%%%%%%%

Then, the following theorem holds.
\begin{theorem}\label{ABC2}
Upon recovery, the maximum achievable entanglement fidelity for the channel of Definition \ref{T2} is:
\begin{itemize}
\item\textbf{Region A}\\
\begin{equation}
F_{max}^{A}(T_{corr})=1-p,\nonumber
\end{equation}
for $0<\mu<\mu_{AB}=\frac{3-4p}{6-4p}$.
\item\textbf{Region B}\\
\begin{equation}
F_{max}^{B}(T_{corr})=(1-\mu)(1-2p+\frac{4p^2}{3})+\mu,\nonumber
\end{equation}
for $\mu>\mu_{AB}$ and $\mu>\mu_{BC}=\frac{4p-3}{4p}$.
\item\textbf{Region C}\\
\begin{equation}
F_{max}^{C}(T_{corr})=(1-\mu)\frac{p}{3}+\mu p,\nonumber
\end{equation}
for $0<\mu<\mu_{BC}$.
\end{itemize}
\end{theorem}

\begin{proof}
The optimal recovery channel is achieved by maximising
expression \eqref{27} over the parameters $q_{_{\a,\gamma}}$.
Our strategy to maximize the entanglement fidelity is to optimize
the correction performance for each channel component
\begin{equation}
T^{(\a)}_{corr}=R_{\a}\circ T_{\a}.\nonumber
\end{equation}
When the outcome of the measurement is $\alpha=0$, the entanglement fidelity of the corrected map $T^{(\a=0)}_{corr}$
is
\begin{eqnarray}
F^{(0)}_{corr}&=&(1-\mu)(1-p)\frac{p}{3}\cr\cr
&+&(1-p)[(1-\mu)(1-\frac{4p}{3})+\mu]q_{_{00}}.\nonumber
\end{eqnarray}
For $(1-\mu)(1-\frac{4p}{3})+\mu>0$ the coefficient of $q_{_{00}}$ is positive, therefore the maximum of $F^{(0)}_{corr}$ is attained by choosing $q_{_{0,0}}=1 $. For
$(1-\mu)(1-\frac{4p}{3})+\mu<0$ the maximum is achieved for
$q_{_{0,0}}=0$. This means that
for $\mu < \frac{4p-3}{4p}$, if the outcome of the measurement is $0$
(no error on the first qubit), the most
appropriate recovery is to perform a Pauli channel on the second qubit and leave the first qubit unchanged.
For $\mu > \frac{4p-3}{4p}$, if the outcome of the measurement is $0$,
the amount of correlation on noise is
large enough to ensure that the second qubit has passed through the channel safely and no correction is
required on either of them.

To find the optimum recovery for the other possible outcomes of the measurement
we have to maximize the following expressions
\begin{equation}\label{fi}
F^{(i)}_{corr}=(1-\mu)\frac{p^2}{9}+(1-\mu)\frac{p}{3}(1-\frac{4p}{3})q_{_{i,0}}+\mu \frac{p}{3} q_{_{i,i}}.
\end{equation}
Notice that the probabilities $q_{_{i,j}}$ with $j\neq i$ do not appear in
(\ref{fi}). Moreover, since $F^{(i)}_{corr}$ is linear in the parameters
$q_{_{i,j}}$ and at least one of the coefficients is positive, remembering
the normalization condition $\sum_{\gamma} q_{_{\a,\gamma}}=1$, we set
$q_{_{i,j}}=0$ for $j\neq i$ to achieve the maximum value for $F^{(i)}_{corr}$.
Hence we can write $q_{_{i,0}}=1-q_{_{i,i}}$. Substituting it in equation
(\ref{fi}) we get
\begin{equation}
F^{(i)}_{corr}=(1-\mu)\frac{p}{3}(1-p)+\frac{p}{3}(\mu-(1-\mu)(1-\frac{4p}{3}))q_{_{i,i}}.\nonumber
\end{equation}
Therefore if the outcome of the measurement is $i=1,2,3$, for $\mu>\frac{3-4p}{6-4p}$ the optimum correction can be
performed by taking $q_{_{i,i}}=1$ and for $\mu<\frac{3-4p}{6-4p}$ the best performance of the recovery is attainable by
taking $q_{_{i,0}}=1$. Therefore the optimum correction varies depending on the values of $p$ and $\mu$, and can be summarized as:

Region A: In this region $0<\mu<\mu_{AB}=\frac{3-4p}{6-4p}$.
The optimum correction is achieved by choosing $q_{_{\a,0}}=1$:
\begin{equation}
R_{\alpha}(\rho')=(\sigma_{\alpha}\otimes  {I})\rho'(\sigma_{\alpha}\otimes  {I}).\nonumber
\end{equation}
Therefore the maximum entanglement fidelity in this region is given by
\begin{equation}
F_{max}^{A}(T_{corr})=1-p.\nonumber
\end{equation}
Region B:
In this region $\mu>\mu_{AB}$ and $\mu>\mu_{BC}=\frac{4p-3}{4p}$.
For the measurement outcome $\a$ the
optimum recovery is given by $q_{_{\a,\a}}=1$:
\begin{equation}
R_{\alpha}(\rho')=(\sigma_{\alpha}\otimes \sigma_{\alpha})\rho'(\sigma_{\alpha}\otimes\sigma_{\alpha}).\nonumber
\end{equation}
Therefore the maximum entanglement fidelity in this region is given by
\begin{equation}
F_{max}^{B}(T_{corr})=(1-\mu)(1-2p+\frac{4p^2}{3})+\mu.\nonumber
\end{equation}
Region C: In this region $0<\mu<\mu_{BC}$.
If the outcome of the measurement is $\a=0$ the
optimal recovery is given by $q_{_{0,0}}=0$:
\begin{equation}
R_{0}(\rho')=\sum_{i=1}^3q_i( {I}\otimes \sigma_{i})\rho'( {I}\otimes \sigma_{i})\hskip 1cm \sum_{i=1}^3q_i=1,
\nonumber
\end{equation}
and for the measurement outcome $i=1,2,3$, the optimal recovery is given by $q_{_{i,i}}=1$
\begin{equation}
R_{i}(\rho')=(\sigma_i\otimes \sigma_{i})\rho'(\sigma_i\otimes \sigma_{i})\hskip 1cm i=1,2,3.\nonumber
\end{equation}
The maximum entanglement fidelity takes the form
\begin{equation}
F_{max}^{C}(T_{corr})=(1-\mu)\frac{p}{3}+\mu p.\nonumber
\end{equation}
\end{proof}
%%%%%%%%%%%%%%%%%%%%%%%%%%%%%%%%%%%

Based on Theorem \ref{ABC2} we can identify three different regions for
the optimum correction in the plane of $p$ and $\mu$, as shown in figure \ref{Regions}.

\begin{remark}
It is interesting to notice that the critical value $\mu_{AB}$ for the
correlation parameter $\mu$ has the same form as the one characterising the
correlated depolarizing channel in terms of classical information transmission \cite{MP}.
In that context the critical value $\mu_{AB}$ gives a threshold value for
the optimal input states: the mutual information along the channel is
maximised with product states for  $\mu\leq\mu_{AB}$, while it achieves its
maximum value with maximally entangled states for $\mu\geq\mu_{AB}$ \cite{MP}.
\end{remark}

%%%%%%%%%%%%%%%%%%%%%%%%%%%%%%%%%%%%%%%%%%

%%%%%%%%%%%%%%%%%%%%%%%%%%%%%%%%%%%%%%

\section{Depolarizing channel for $n$ qubits}

In the previous section we have seen how the performance of the
quantum feedback control scheme behaves if we can perform measurements over
$\mathcal{L}(\mathcal{K})$ while the total
Hilbert space of the environment is $\mathcal{K}\otimes \mathcal{K}$.
If we had full access to the environment
we could completely retrieve quantum information.
However, we have shown that having only partial access to the environment
and exploiting the correlation in
noise we are still capable of partially recovering quantum information.
Now we want to see how the performance of
the correction behaves if we keep increasing the Hilbert space of
the environment without increasing our
 access to it. To do so we consider a correlated depolarizing channel
 defined by $T:\mathcal{L}(\mathcal{H}^{\otimes n})\rightarrow \mathcal{L}
(\mathcal{H}^{\otimes n})$, resembling the long term memory channels
introduced in \cite{DD07},
and we perform measurement on $\mathcal{L}(\mathcal{K})$ while the state of the environment belongs to $\mathcal{L}(\mathcal{K}^{\otimes n})$.

\begin{definition}\label{Tn}
Let us define the selected output of the channel corresponding to the classical outcome of the measurement, $\a$,
as
\begin{equation}
T_{\a}(\rho):=(1-\mu)T_{\a}^{uc}(\rho)+\mu T_{\a}^{cc}(\rho),\nonumber
\end{equation}
with
\begin{equation}
T_{\a}^{uc}(\rho):=\sum_{\b_1\cdots\b_n}p_{\a}\prod_{i=2}^np_{_{\b_i}}(\sigma_a\otimes\bigotimes_{i=2}^n\sigma_{_{\b_i}})
\rho(\sigma_a\otimes\bigotimes_{i=2}^n\sigma_{_{\b_i}}),\nonumber
\end{equation}
and
\begin{equation}
T_{\a}^{cc}(\rho):=p_{\a}\sigma_{_{\a}}^{\otimes n}\rho\sigma_{_{\a}}^{\otimes n}.\nonumber
\end{equation}
\end{definition}\medskip

Similarly to the previous section we have the following

\begin{lemma}\label{recn}
The recovery operator
\begin{equation}
R_{\a}(\rho')=\sum_{\gamma_2,\cdots,\gamma_n}q_{_{\a,\gamma_2,\cdots,\gamma_n}}(\sigma_a\otimes\bigotimes_{i=2}^n\sigma_{_{\gamma_i}})\rho(\sigma_a\otimes\bigotimes_{i=2}^n\sigma_{_{\gamma_i}}),\nonumber
\end{equation}
with the constraint
\begin{equation}\label{qn}
\sum_{\gamma_2,\cdots,\gamma_n}q_{_{\a,\gamma_2,\cdots,\gamma_n}}=1,\hskip 5mm \forall\a,
\end{equation}
is optimal for the channel of Definition \ref{Tn}.
\end{lemma}

\begin{proof}
Having the classical outcome of the measurement $\a$ we know that error $\sigma_{\a}$ has occurred on the first
qubit and the effect
of error can be completely removed by performing $\sigma_{\a}$ on the first qubit for correction. Therefore the
recovery map should have the following form
\begin{equation}
R_{\a}(\rho')=\sum_{\gamma_2,\cdots,\gamma_n}(\sigma_{\a}\otimes \mathcal{A}^{(\a)}_{_{\gamma_2,\cdots,\gamma_n}})\rho' (\sigma_{\a}\otimes \mathcal{A}^{(\a)}_{_{\gamma_2,\cdots,\gamma_n}})^{\dag}
\end{equation}
where $\mathcal{A}^{(\a)}_{_{\gamma_2,\cdots,\gamma_n}}$ is an operator
acting on $n-1$ qubits. Expanding it in terms of products of Pauli matrices
we have
\begin{equation}
\mathcal{A}^{(\a)}_{_{\gamma_2,\cdots,\gamma_n}}=\sum_{_{\delta_2,\cdots,\delta_n}}c^{^{\gamma_2,\cdots,\gamma_n}}_{_{\a,\delta_2,\cdots,\delta_n}}\sigma_{\delta_2}\otimes\sigma_{\delta_3}\otimes\cdots\otimes\sigma_{\delta_n}
\end{equation}
The completeness condition
\begin{equation}
\sum_{\gamma_2,\cdots,\gamma_n}(\sigma_{\a}\otimes \mathcal{A}^{(\a)}_{_{\gamma_2,\cdots,\gamma_n}})^{\dag}(\sigma_{\a}\otimes \mathcal{A}^{(\a)}_{_{\gamma_2,\cdots,\gamma_n}})=I,\nonumber
\end{equation}
imposes the following constraint on the
coefficients $c^{^{\gamma_2,\cdots,\gamma_n}}_{_{\a,\delta_2,\cdots,\delta_n}}$
\begin{equation}
\sum_{\gamma_2,\cdots,\gamma_n}|c^{^{\gamma_2,\cdots,\gamma_n}}_{_{\a,\delta_2,\cdots,\delta_n}}|^2=1,
\end{equation}
Considering this general form of Kraus operators for the recovery map and using (\ref{ftc}), the entanglement fidelity of the corrected channel
takes the form
\begin{eqnarray}
F(T_{corr})&=&(1-\mu)\sum_{_{\a,\beta_2,\cdots\beta_n}}\sum_{_{\gamma_2\cdots\gamma_n}}|c^{^{\gamma_2,\cdots,\gamma_n}}_{_{\a,\beta_2,\cdots,\beta_n}}|^2p_{\a}\prod_{i=2}^np_{_{\beta_i}}\cr\cr\cr
&+&\mu \sum_{\a,\gamma_2,\cdots,\gamma_n}|c^{^{\gamma_2,\cdots,\gamma_n}}_{_{\a,\a,\cdots,\a}}|^2p_{\a}.
\end{eqnarray}
The above expression can be written as
\begin{eqnarray}\label{entfnc}
F(T_{corr})&=&(1-\mu)\sum_{_{\a,\beta_2,\cdots\beta_n}}p_{\a}
q_{_{\a,\beta_2,\cdots,\beta_n}}\prod_{i=2}^np_{_{\beta_i}}\cr\cr
&+&\mu \sum_{\a}p_{\a}q_{_{\a,\a,\cdots,\a}}.
\end{eqnarray}
By defining
\begin{equation}
q_{_{\a,\beta_2,\cdots,\beta_n}}=\sum_{_{\gamma_2\cdots\gamma_n}}|c^{^{\gamma_2,\cdots,\gamma_n}}_{_{\a,\beta_2,\cdots,\beta_n}}|^2
\end{equation}
The same value of entanglement fidelity in equation \eqref{entfnc} will be obtained by assuming
the recovery map with following Kraus operators:
\begin{equation}
\sqrt{q_{_{\a,\gamma_2,\cdots,\gamma_n}}}\sigma_{\a}\otimes\sigma_{\gamma_2}\otimes\cdots\otimes\sigma_{\gamma_n}
\end{equation}
with the constraint that $\sum_{_{\gamma_2\cdots\gamma_n}}q_{_{\a,\gamma_2,\cdots,\gamma_n}}=1$ for all $\a$.
\end{proof}

The entanglement fidelity \eqref{ftc} of the channel in the Definition \ref{Tn} corrected using Theorem \ref{recn} results
\begin{eqnarray}
F(T_{corr})&=&(1-\mu)\sum_{\a,\b_2\cdots\b_n}p_{\a}p_{_{\b_2}}\cdots p_{_{\b_{n}}}q_{_{\a,\b_2,\cdots,\b_n}}\cr\cr
&+&\mu\sum_{\a}p_{\a}q_{_{\a,\a,\cdots,\a}}.
\label{Fnqadd}
\end{eqnarray}
%%%%%%%%%%%%%%%%%%%%%%%%%%%%%%%%%%%%%%%
\begin{figure}[t]
\centering
\includegraphics[height=0.25\textheight]{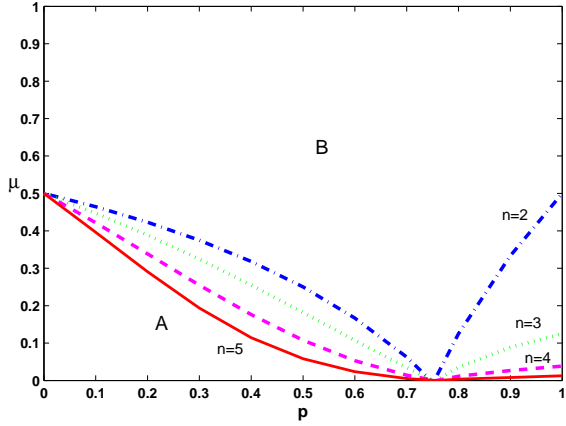}
\caption{Different parameters regions for optimal recovery in the case of two (dashed-dotted line), three (dotted line), four (dashed line) and five (solid line) qubit channels.} \label{RegionsNqubits}
\end{figure}
%%%%%%%%%%%%%%%%%%%%%%%%%%%%%%%%%%%%%%%%%%%%%%%%
Then we have the following Theorem.
\begin{theorem}\label{ABCn}
Upon recovery, the maximum achievable entanglement fidelity for the channel of Definition \ref{Tn} is:
\begin{itemize}
\item\textbf{Region A}\\
\begin{equation}
F_{max}^{A}(T_{corr})=(1-\mu)(1-p)^{n-1}+\mu(1-p),\nonumber
\end{equation}
for $0<\mu<\frac{(1-p)^{n-1}-(\frac{p}{3})^{n-1}}{(1-p)^{n-1}-(\frac{p}{3})^{n-1}+1}$.
\item\textbf{Region B}\\
\begin{equation}
F_{max}^{B}(T_{corr})=(1-\mu)((1-p)^{n}+3(\frac{p}{3})^n)+\mu,\nonumber
\end{equation}
for $\mu>\frac{(1-p)^{n-1}-(\frac{p}{3})^{n-1}}{(1-p)^{n-1}-(\frac{p}{3})^{n-1}+1}$ and  $\mu>-\frac{(1-p)^{n-1}-(\frac{p}{3})^{n-1}}{(1-p)^{n-1}-(\frac{p}{3})^{n-1}-1}$.
\item\textbf{Region C}\\
\begin{equation}
F_{max}^{C}(T_{corr})=(1-\mu)(\frac{p}{3})^{n-1}+\mu p,\nonumber
\end{equation}
for $0<\mu<-\frac{(1-p)^{n-1}-(\frac{p}{3})^{n-1}}{(1-p)^{n-1}-(\frac{p}{3})^{n-1}-1}$.
\end{itemize}
\end{theorem}

\begin{proof}
To maximize the entanglement fidelity in \eqref{Fnqadd} over
its parameters, we maximize it for each value of the measurement outcome
$\a$. If the outcome of the measurement is zero then the entanglement fidelity is given by
\begin{eqnarray}\label{F0napp1}
F^{(0)}_{corr}&=&[(1-\mu)(1-p)^n+\mu (1-p)]q_{{_0,\cdots,_0}}\cr\cr
&+&(1-\mu)(1-p)^{n-1}\frac{p}{3}\sum_{j,perm}q_{_0,_j,_0
\cdots,_0}\cr\cr
&+&(1-\mu)(1-p)^{n-2}(\frac{p}{3})^{2}\sum_{j,k,perm}q_{_0,_j,_k,_0,\cdots,_0}
+\cdots\cr\cr
&+&(1-\mu)(1-p)(\frac{p}{3})^{n-1}\sum_{i_2\cdots i_n}q_{_0,i_2,\cdots,i_n},
\end{eqnarray}
where the notation $perm$ in the summations above refers to all possible
permutations of the last $n-1$ indexes. Notice that for $p<3/4$ the
largest coefficient in the above expression is the one in front of
$q_{_0,_0,_0 \cdots,_0}$ and therefore in this case, for any value of $\mu$,
$F^{(0)}_{corr}$ is always maximised by $q_{_0,_0,_0 \cdots,_0}=1$. In the case
of $p>3/4$ the largest coefficient, excluding the first,
is given by the last one in (\ref{F0napp1}). Therefore, in this case the
optimal solution can be searched by setting to zero all values of $q$ that
contain at least one value 0 among the last $n-1$ indexes.
The expression for the entanglement fidelity that we are going to maximise is
now simplified as
\begin{eqnarray}\label{F0napp}
F^{(0)}_{corr}&=&[(1-\mu)(1-p)^n+\mu (1-p)]q_{_{0,\cdots,0}}\cr\cr
&+&(1-\mu)(1-p)(\frac{p}{3})^{n-1}\sum_{i_2\cdots i_n}q_{_0,i_2,\cdots,i_n}.
\end{eqnarray}
Using the condition in equation (\ref{qn}) we know that
\begin{equation}\label{conditionQ1}
\sum_{i_2\cdots i_n}q_{_0,i_2,\cdots,i_n}=1-q_{_{0,\cdots,0}}.\nonumber
\end{equation}
Replacing it in equation (\ref{F0napp}) we find that
\begin{eqnarray}\label{F0n}
&F^{(0)}_{corr}=(1-\mu)(1-p)(\frac{p}{3})^{n-1}\nonumber\\
&+\left((1-\mu)[(1-p)^n-(1-p)(\frac{p}{3})^{n-1}]+\mu (1-p)\right)q_{_{0,\cdots,0}}.\nonumber
\end{eqnarray}
It is easy to see that for
\begin{equation}
0<\mu<-\frac{(1-p)^{n-1}-(\frac{p}{3})^{n-1}}
{(1-p)^{n-1}+(\frac{p}{3})^{n-1}-1}\nonumber
\end{equation}
the coefficient of $q_{_{0,\cdots,0}}$ is negative,
therefore the maximum of $F^{(0)}_{corr}$ is attainable for $q_{_{0,\cdots,0}}=0$. For
\begin{equation}
\mu>-\frac{(1-p)^{n-1}-(\frac{p}{3})^{n-1}}{(1-p)^{n-1}+
(\frac{p}{3})^{n-1}-1}\nonumber
\end{equation}
the coefficient of $q_{_{0,\cdots,0}}$ is positive so the maximum of
$F^{(0)}_{corr}$ is achieved by taking
$q_{_{0,\cdots,0}}=1$.

When the outcome of the measurement is $i=1,2,3$ then the entanglement
fidelity of the corrected channel is
\begin{eqnarray}\label{Finapp1}
&&F^{(i)}_{corr}=(1-\mu)(1-p)^{n-1}\frac{p}{3}q_{{_i,_0,\cdots,_0}}\cr\cr
&&+(1-\mu)(1-p)^{n-2}(\frac{p}{3})^{2}\sum_{j,perm}q_{_i,_j,_0
\cdots,_0}\cr\cr
&&+(1-\mu)(1-p)^{n-3}(\frac{p}{3})^{3}\sum_{j,k,perm}q_{_i,_j,_k,_0,\cdots,_0}
+\cdots\cr\cr
&&+\frac{p}{3}\sum_{i_2\cdots i_n}q_{i,i_2,\cdots,i_n}\left[(1-\mu)
(\frac{p}{3})^{n-1}+\mu\prod_{i_k,_{k=2}}^n \delta_{i,i_k}
\right].\nonumber\\
\end{eqnarray}
Notice that for $p>3/4$ the
largest coefficient in the above expression is the one in front of
$q_{i,i,\cdots,i}$, therefore in this case the optimal solution corresponds
to $q_{i,i,\cdots,i}=1$ for any value of $\mu$. Notice also that in the
last line the coefficient in front of $q_{i,i,\cdots,i}$ is always larger
than the other ones, so in order to look for the maximum we can always
set $q_{i,i_2,\cdots,i_n}=0$ for all cases except $q_{i,i,\cdots,i}$.
Moreover, for  $p<3/4$, the largest coefficients in  (\ref{Finapp1}),
with the exclusion of the last line, is the one in front of
$q_{{_i,_0,\cdots,_0}}$. Therefore, by these considerations, we can restrict
our search for the maximum values to the case of vanishing $q$ except for
$q_{{_i,_0,\cdots,_0}}$ and $q_{i,i,\cdots,i}$. We can then write
\begin{eqnarray}
F^{(i)}(T_{corr})&=&(1-\mu)(1-p)^{n-1}\frac{p}{3}q_{_{_i,_0,\cdots,_0}}\cr\cr
&+&[(1-\mu)(\frac{p}{3})^n+\mu \frac{p}{3}] q_{_{i,i,\cdots,i}}.\nonumber
\end{eqnarray}
From equation (\ref{qn}) we know that
\begin{equation}\label{conditionQ2}
q_{_{i,i,\cdots,i}}=1-q_{_{i,0,\cdots,0}}.\nonumber
\end{equation}
Therefore
\begin{eqnarray}
&&F^{(i)}(T_{corr})=(1-\mu)(\frac{p}{3})^n+\mu\frac{p}{3}\cr\cr
&+&\frac{p}{3}\left((1-\mu)[(1-p)^{n-1}-(\frac{p}{3})^{n-1}]-\mu\right) q_{_{i,0,\cdots,0}}.\nonumber
\end{eqnarray}
The coefficient of $q_{_{i,0,\cdots,0}}$ is positive for
\begin{equation}
0<\mu<\frac{(1-p)^{n-1}-(\frac{p}{3})^{n-1}}{(1-p)^{n-1}+(\frac{p}{3})^{n-1}+1}.\nonumber
\end{equation}
Therefore in this region we should take  $q_{_{i,0,\cdots,0}}=1$ and for
\begin{equation}
\mu>\frac{(1-p)^{n-1}-(\frac{p}{3})^{n-1}}{(1-p)^{n-1}+(\frac{p}{3})^{n-1}+1},\nonumber
\end{equation}
the coefficient of $q_{_{i,0,\cdots,0}}$ is negative and therefore we should take $q_{_{i,i,\cdots,i}}=1$.
\end{proof}\medskip

%%%%%%%%%%%%%%%%%%%%%%%%%%%%%%%%%%%%%%%

Figure \ref{RegionsNqubits} shows the regions $A$, $B$, $C$
of Theorem \ref{ABCn} for different values of $n$ (including the previous analyzed case of $n=2$). We can see that
by increasing $n$ the regions $A$ and $C$ become smaller. This can be
understood by noticing that
the region $A$ corresponds to the case where although the
measurement outcome shows that an error has occurred on the first qubit,
we expect that the other qubits have not experienced any error. However,
the chance that this holds true is lowered by increasing $n$. The same reasoning can be applied to the region $C$.

Figure \ref{FvN} shows the entanglement fidelity versus the number of qubits in the system for $p=0.4$ and for different values of
$\mu$. It is interesting
to notice that for large $n$ the entanglement fidelity does not go zero, due
to the role of noise correlation in performing the recovery operation.
%%%%%%%%%%%%%%%%%%%%%%%%%%%%%%%%%%%%%%%
\begin{figure}[t]
\centering
\includegraphics[height=0.25\textheight]{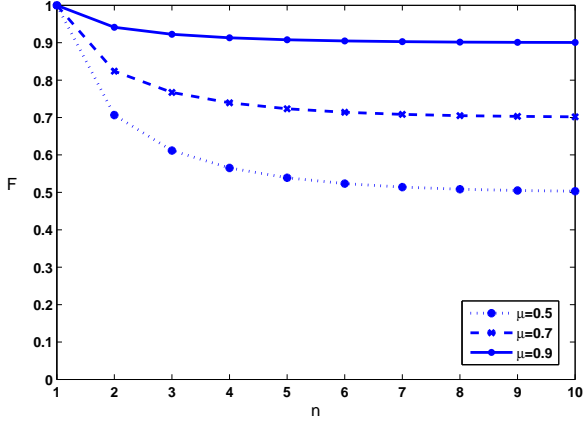}
\caption{Entanglement fidelity versus $n$ for $p=0.4$ and $\mu=0.9$ ( solid line), $\mu=0.7$ (dashed line) and $\mu=0.5$ (dotted line).} \label{FvN}
\end{figure}
%%%%%%%%%%%%%%%%%%%%%%%%%%%%%%%%%%%%%%%%%%%%%%%%

\section{Conclusions}

The main result presented in this paper is the possibility of recovering
quantum information on a multipartite system by using limited access to the
environment.
In particular, we have addressed the important question of how well quantum
information can be recovered on a multiple qubit system by performing
a measurement on one environmental subsystem.
We have considered a quite general kind of correlated errors on qubits and we have determined the optimal recovery depending on the degree of errors correlation.
We have also found the scaling of the performance versus the number of qubits while monitoring the error just on one of them.
Interestingly enough, for finite degree of errors correlation, the recovery ability is preserved by increasing the number of non-measured subsystems of the environment.

As a final remark, we point out that when considering partial control one could exploit correlations residing on the system's state itself rather than on the errors.  This would be more in the spirit of Refs.\cite{man} and is left for future investigations.

%\appendices

\section*{Acknowledgment}

We acknowledge the financial support of the European Commission, under the FET-Open grant agreement CORNER, number FP7-ICT-213681.

\ifCLASSOPTIONcaptionsoff
  \newpage
\fi

\end{document}